%% file: main.tex
\documentclass{article}
\input{macros}

\usepackage{graphicx}

\title{
Optimal Online Discrepancy Minimization in Linear Time}
\author{Ishaq Aden-Ali \thanks{Unemployed. Email: {adenali@berkeley.edu}}}
\date{June 27, 2026}

\begin{document}

\maketitle
\begin{abstract}
We provide an online algorithm with the following guarantee: for any fixed sequence of vectors $v_1,\dots,v_T \in \mb{R}^d$ with $\|v_i\|_2\le 1$, the algorithm assigns each arriving vector $v_t$ a random sign $\varepsilon_t$ such that every prefix sum $\sum_{i=1}^t \varepsilon_i v_i $
can be written as the sum of three coupled standard Gaussian vectors.
Our algorithm runs in $O(dT)$ time
and achieves the optimal prefix discrepancy bound
\[
\max_{1 \le t \le T}\left\| \sum_{i=1}^t \eps_i v_i \right\|_\infty = O\left( \sqrt{\log T} \right),
\]
with high probability.
This recovers the optimal bound of Kulkarni, Reis, and Rothvoss, whose algorithm runs in time exponential in $T$ and $d$.
The algorithm and main proof were discovered in a GPT-5.5 Pro Extended conversation prompted by the author. 
\end{abstract}

\paragraph{Statement on AI use.} The algorithm and main proof were discovered by OpenAI's GPT-5.5 Pro Extended reasoning model.
The proof appeared in a discussion that began with the author prompting the model to try to identify proof barriers to a natural family of algorithms related to the self-balancing walk of Alweiss, Liu, and Sawhney \cite{AlweissLS21}.
After a few rounds of prompting, a number of barriers emerged. 
The author then prompted the model to search for an improved algorithm that could overcome these barriers.
After a few more rounds of prompting, the model proposed the algorithm below, which is based on coupling three Gaussian fixed point walks, and provided a self-contained proof of the main theorem.
The author subsequently verified and rewrote the main proof, derived the straightforward implications to discrepancy minimization, and wrote the paper. GPT-5.5 Pro was also used to help find typographical errors and search for related work.
\section{Introduction}
We study online discrepancy minimization problems against an oblivious adversary: The adversary generates a sequence of vectors $v_1, v_2, \dots, v_T \in \mb{R}^d$ upfront and reveals them to the player one at a time.
Upon receiving $v_t$ at time step $t \in [T] \coloneqq \{1, \dots, T\}$, the player must commit to a sign $\eps_t \in \{\pm 1\}$ before receiving the next vector $v_{t+1}$.
The player's (resp.\ adversary's) goal is to minimize (resp.\ maximize) the \emph{prefix discrepancy}
\begin{equation}\label{eq:discrepancy}
\max_{t \in [T]}\left\|\sum_{i =1}^t \eps_i v_i \right\|_\infty.
\end{equation}
Spencer~\cite{Spencer77} was the first to study a more general version of this problem.
A number of papers study this online discrepancy problem with various distributional assumptions on the input~\cite{BansalS20,BansalJ0S20,BansalJMSS21,AltschulerT25}.
We focus on the worst-case setting where we are only assume $\|v_t\|_2 \le 1$. We refer to this as the \emph{online Koml\'os problem}, after the (offline) Koml\'os problem.\footnote{Koml\'os conjectured that for any collection of vectors $v_1, \dots, v_T$ with $\max_{i \in [T]} \|v_i\|_2 \le 1$, there are signs $\{\eps_i\}_{i=1}^T$ such that $\|\sum_{i=1}^T \eps_i v_i\|_\infty$ is constant.}

A sequence of works has studied the online Koml\'os problem \cite{AlweissLS21,LiuSS22,KulkarniRR24}. 
A common approach in this line of work is to design a randomized online algorithm that generates random signs $\eps_1, \eps_2, \dots, \eps_T$ such that every prefix sum $S_t = \sum_{i=1}^t \eps_i v_i$ is $K$-subgaussian, i.e., for any fixed vector $\theta \in \mb{R}^d$,
\[
 \E \left[\exp\left( \la S_t, \theta \ra \right) \right] \le \exp\left( \frac{K^2\|\theta\|_2^2}{2} \right).
\]
It then follows from standard concentration arguments that the prefix discrepancy does not exceed $O(K\sqrt{ \log T})$ with high probability.\footnote{The implied bound is often written $O(K\sqrt{\log dT})$ but a slightly more careful union bound removes the $K\sqrt{\log d}$ term.}
Alweiss, Liu, and Sawhney~\cite{AlweissLS21} introduced the elegant \emph{self-balancing walk} that keeps every prefix sum $\Theta(\sqrt{\log T})$-subgaussian, yielding a prefix discrepancy bound of $O(\log T)$. 
Their algorithm runs in linear time.
Later, Kulkarni, Reis, and Rothvoss~\cite{KulkarniRR24} were able to construct a distribution that keeps every prefix sum $C$-subgaussian for a universal constant $C$. This allowed them to obtain an $O(\sqrt{\log T})$ bound on the prefix discrepancy, and they showed that this bound is optimal.
Unfortunately, their algorithm runs in time exponential in $T$ and $d$.

An intermediate result of Liu, Sah, and Sawhney~\cite{LiuSS22} is conceptually the most relevant to our work.
These authors designed a particular \emph{Gaussian fixed point walk} in the online setting that we consider.
At time step $t \in [T]$, this process samples a random \emph{fractional} sign $\delta_t \in \{-1, 0 , 1\}$ so that the prefix sum satisfies
\[
\sum_{i=1}^t \delta_i v_i = -G_0 + G_{t},
\]
where $G_0$ and $G_{t}$ are both \emph{marginally} distributed as a mean zero Gaussian with covariance $\sigma^2 I_d$, i.e.\ $G_0 \sim \mc{N}(0, \sigma^2 I_d)$ and $G_t\sim \mc{N}(0, \sigma^2 I_d)$.\footnote{They also construct a version of this walk where $\delta_t \in \{-1, 1, 2\}$.}
More specifically, their random walk defines a Markov chain on $\mb{R}^d$ whose transition rule depends on the incoming vector $v_t$: given the current state $G_{t-1}$, the chain assigns transition probabilities to the three possible next states $\{G_{t-1}-v_t, G_{t-1}, G_{t-1}+v_t\}$ so that $\mc{N}(0,\sigma^2 I_d)$ is invariant.
The walk starts at $G_0 \sim \mc{N}(0,\sigma^2 I_d)$, so $G_t\sim \mc{N}(0,\sigma^2 I_d)$ throughout.
While it is tempting to try to improve these Markov chains so that they always output signs, the authors point out that this is not possible.
Still, they show that by setting the variance of the invariant Gaussian to $\Theta(\log T)$, then with very high probability every $\delta_t$ that is sampled is a sign.
 Since this choice of variance implies $O(\sqrt{\log T})$-subgaussian prefix sums, their algorithm achieves an $O(\log T)$ prefix discrepancy bound matching the earlier work \cite{AlweissLS21}.

The starting point of our work is a different, more ``balanced", Gaussian fixed point walk.
Unlike the walk described above, our walk never assigns too much probability mass to any of the three possible fractional signs.
This allows us to couple three copies of the Gaussian fixed point walk in a way that preserves the marginal distribution of each update $\delta_{t,j}$, while enforcing $\eps_t = \delta_{t, 1} +  \delta_{t, 2}  +  \delta_{t, 3} \in \{ \pm 1 \}$.
Since we can initialize the three walks so that $\sum_{j=1}^3 G_{0, j} = 0$, adding up their updates yields a signed walk with prefix sum
\[
S_t = \sum_{i=1}^t \eps_i v_i = G_{t, 1} + G_{t, 2} + G_{t, 3},
\]
where each $G_{t, j} \sim \mc{N}(0, I_d)$.
Thus, every prefix sum $S_t$ is $3$-subgaussian by H\"older's inequality.
We dub this algorithm the Gaussian triplet walk and describe it formally in \cref{alg:online-gaussian-coupled-coloring}.
Our main theorem is the following. 
\begin{theorem}\label{thm:3gaussians}
Fix a sequence $v_1, \dots, v_T \in \mb{R}^d$ such that $\|v_i\|_2 \le 1$ for all $i \in [T]$. 
At time step $t \in [T]$, \cref{alg:online-gaussian-coupled-coloring} receives $v_t$ and chooses a random sign $\eps_t \in \{ \pm 1\}$ such that the prefix sum can be written as
\[
\sum_{i=1}^t \eps_i v_i = G_{t,1} + G_{t,2} + G_{t,3},
\]
where $G_{t,j}$ has marginal distribution $\mc{N}(0, I_d)$ for every $j \in [3]$.
Moreover, \cref{alg:online-gaussian-coupled-coloring} runs in $O(dT)$ time.
\end{theorem}
The following online discrepancy bound is an immediate consequence of \cref{thm:3gaussians}.
\begin{theorem}\label{thm:prefix_discrepancy}
Fix a sequence $v_1, \dots, v_T \in \mb{R}^d$ such that $\|v_t\|_2 \le 1$ for all $t \in [T]$.
At time step $t \in [T]$, \cref{alg:online-gaussian-coupled-coloring} receives $v_t$ and chooses a random sign $\eps_t \in \{ \pm 1\}$.
For any $\delta \in (0, 1/2)$ it holds, with probability at least $1-\delta$, that
\[
 \max_{1 \le t \le T }\left\|\sum_{i=1}^t \eps_i v_i \right\|_\infty \le \sqrt{18\log \left(\frac{2T^2}{\delta}\right)}.
\]
Moreover, \cref{alg:online-gaussian-coupled-coloring} runs in $O(dT)$ time.
\end{theorem}
Previously, the prefix discrepancy bound in \cref{thm:prefix_discrepancy} was only known to be achievable by an algorithm running in time exponential in $T$ and $d$ \cite{KulkarniRR24}.
In fact, a more general version of \cref{thm:prefix_discrepancy} holds for a large family of norms, but we choose to focus on the $\ell_\infty$ norm for simplicity; see \cite[Theorem 1.2]{KulkarniRR24} for details.
The guarantee of~\cref{thm:prefix_discrepancy} is essentially optimal given the lower bound \cite[Theorem 1.4]{KulkarniRR24}.

In the offline setting, our algorithm gives a linear-time algorithmic version of Banaszczyk's celebrated result \cite{Banaszczyk98} for the Koml\'os problem.
\begin{corollary}\label{cor:offline_komlos}
Let $A \in \mb{R}^{d \times T}$ be a matrix whose columns have $\ell_2$ norm at most $1$. 
There is a randomized algorithm that takes $A$ as input and outputs signs $\eps \in \{\pm 1\}^T$ such that, for any $\delta \in (0, 1/2)$ it holds, with probability at least $1-\delta$, that
\[
\|A\eps\|_\infty \le  \sqrt{18\log\left(\frac{2\min\{d,T\}}{\delta}\right)}.
\]
Moreover, the algorithm runs in $O(dT)$ time.
\end{corollary}
Previous algorithmic versions of Banaszczyk's bound gave a number of polynomial-time algorithms with different runtime tradeoffs~\cite{LevyRR17, DadushGLN19,BansalDG19,BansalDGL19,BansalLV22,PesentiV23}.
Finally, we note that Banaszczyk's bound for the Koml\'os problem is suboptimal in light of the recent breakthrough by Bansal and Jiang~\cite{BansalJ26}, who achieve the bound $\widetilde{O}(\log^{1/4} T)$ with an efficient algorithm.

\section{Gaussian triplet walk}\label{sec:gaussian_triplet_walk}
In this section, we analyze the Gaussian triplet walk.
We first introduce a useful Gaussian fixed point walk.
For $x, v \in \mb{R}^d$, define the following three probabilities
\begin{align*}\label{eq:markov_chain}
\begin{split}
    p_v^+(x)
    &=\frac13\min\left\{1,
      e^{-\langle x,v\rangle-\|v\|_2^2/2}\right\},\\
    p_v^-(x)
    &=\frac13\min\left\{1,
      e^{\langle x,v\rangle-\|v\|_2^2/2}\right\},\\
    p_v^0(x)&=1-p_v^+(x)-p_v^-(x).
\end{split}
\end{align*}
Consider the random walk on $\mb{R}^d$ defined by the following time-dependent Markov chain. At time $t$, the chain transitions from state $G_{t-1}$ to one of the three states $G_{t-1} - v_t$,  $G_{t-1}$, or $ G_{t-1} + v_t$ with probabilities $p_{v_t}^-(G_{t-1})$, $p_{v_t}^0(G_{t-1})$, and $p_{v_t}^+(G_{t-1})$ respectively.
This walk preserves $\mc{N}(0, I_d)$ at every time step.
\begin{proposition}\label{prop:gaussian_markov_chain}
Fix $v \in \mb{R}^d$ and let $G \sim \mc{N}(0, I_d)$.
Let $\delta \in \{-1, 0, 1\}$ be the random variable defined by
\[
\Pr(\delta = -1  \mid G ) = p_v^-(G), \quad  \Pr(\delta = 1  \mid G) = p_v^+(G), \quad \Pr(\delta = 0  \mid G) = p_v^0(G). 
\]
Then $G + \delta v \sim \mc{N}(0, I_d)$.
\end{proposition}
\begin{proof}
Let $f$ denote the density of $G+ \delta v$ and write $\gamma$ for the density of the standard Gaussian. 
Since $G \sim \mc{N}(0, I_d)$, the density of $G+ \delta v$ at any point $x$ is
\[
\begin{aligned}
    f(x)
    &=
    \gamma(x-v)p_{v}^+(x-v)
    +
    \gamma(x+v)p_{v}^-(x+v)
    +
    \gamma(x)p_{v}^0(x).
\end{aligned}
\]
A Gaussian density-ratio calculation allows to rewrite
\[
    p_{v}^+(x)
    =
    \frac13\min\left\{1,\frac{\gamma(x+v)}{\gamma(x)}\right\},
    \qquad
    p_{v}^-(x)
    =
    \frac13\min\left\{1,\frac{\gamma(x-v)}{\gamma(x)}\right\}.
\]
Thus,
\[
    \gamma(x-v)p_{v}^+(x-v) = \frac13\min\{\gamma(x-v),\gamma(x)\} = \gamma(x)p_{v}^-(x),
\]
and similarly
\[
    \gamma(x+v)p_{v}^-(x+v) = \gamma(x)p_{v}^+(x).
\]
Substituting these two identities into the expression for $f(x)$ above yields
\[
\begin{aligned}
    f(x)=
    \gamma(x)p_{v}^-(x)
    +
    \gamma(x)p_{v}^+(x)
    +
    \gamma(x)p_{v}^0(x)  =\gamma(x),
\end{aligned}
\]
so $G + \delta v \sim\mc{N}(0,I_d)$ as desired.
\end{proof}
Our Gaussian fixed point walk does not assign too much probability mass to any single next state.  
\begin{proposition}\label{prop:nice_probs}
If $\|v\|_2\leq1$, then for every $x\in\mb{R}^d$,
\[
    p_v^+(x),p_v^-(x)\leq\frac13,
    \qquad
    p_v^+(x) + p_v^-(x) \ge \frac13.
\]
\end{proposition}
\begin{proof}
The upper bounds follow by definition.
Set
$z =\langle x,v\rangle$ and $r =\|v\|_2^2/2\leq1/2$.  If
$z\geq r$, then $p_v^-(x)=1/3$.
Similarly, if $z\leq-r$, then
$p_v^+(x)=1/3$.  In the remaining case $z  \in (-r, r)$,
\[
    p_v^+(x)+p_v^-(x)
    =\frac23e^{-r}\cosh z
    \geq \frac23e^{-1/2}>\frac13.\qedhere
\]
\end{proof}
Next, we show that, under the balance condition, there is a coupling on $\{-1, 0, 1\}^3$ with the correct marginals whose coordinates sum to either $-1$ or $1$.
\begin{lemma}\label{lem:three-way-coupling}
Given numbers $(a_j, b_j)_{j=1}^3$ satisfying $ 0 \le a_j, b_j \le \frac{1}{3}$ and $a_j+b_j \ge \frac{1}{3}$ as input, \cref{alg:three-way-coupling} outputs a random vector $\delta \in \{-1, 0, 1\}^3$ with $\sum_{j=1}^3 \delta_j  \in \{\pm 1\}$ deterministically and  marginal distributions
\[
    \Pr\left( \delta_j = 1 \right) = a_j, \quad \Pr\left(\delta_j = -1\right) = b_j.
\]
\end{lemma}
\begin{proof}
We begin by showing that~\cref{alg:three-way-coupling} can find a point $x \in [0,1]^3$ satisfying
\[
\max\{\tau-b_j,0\} \le x_j \le \min\{\tau, a_j\}, \quad \text{and} \quad \tau \le \sum_{j=1}^3 x_j \le 2 \tau.
\]
Set $s_j = a_j + b_j$.
By assumption, $1/3 \le s_j \le 2/3$, so we have $0\le\tau\le1/2$.  Similarly, 
\[
    2\tau
    =s_1+s_2+s_3-1
    \le s_j+\frac13
    \le2s_j,
\]
so $\tau\le s_j$ for every $j \in [3]$.
Let $u_j = \min\{\tau, a_j\}$ and $l_j = \max\{\tau-b_j,0\}$ be the upper and lower bounds on the $j$th coordinate $x_j$. 
Since $\tau \le s_j = a_j + b_j$, we have $l_j \le u_j$ for every $j \in [3]$.
Thus, the box $ \prod_{j=1}^3 [l_j , u_j] \subseteq [0,1]^3$ is well defined.
We now show how to greedily find a point $x$ in the box satisfying the condition
\begin{equation}
    \tau\le\sum_{j=1}^3x_j\le2\tau.
    \label{eq:x-constraints}
\end{equation}
Set $A = \sum_{j=1}^3 a_j$ and $B = \sum_{j=1}^3 b_j$. We have $A, B \le 1$ and $A+B= 2\tau + 1$, so $A, B \ge \tau$. Thus,
\[
\sum_{j=1}^3 u_j\ge\tau,
    \qquad
    \sum_{j=1}^3 l_j
    =3\tau-\sum_{j=1}^3\min\{\tau,b_j\}
    \le2\tau.
\]
Start the greedy process at $x =(l_1, l_2, l_3)$. If $\sum_{j=1}^3 l_j \ge \tau$, we are done.
Otherwise, increase the coordinates of $x$ one-by-one up to $u_j$ until the sum reaches $\tau$. This terminates at a valid point because $\sum_{j=1}^3 u_j \ge \tau$.

Let $U$  be the uniform random variable on $[0,1)$ sampled by \cref{alg:three-way-coupling}.
We first analyze the branch based on the event $E_1 = \{U < \tau\}$.
Notice that conditioned on $E_1$, $U' \coloneqq U/\tau$ is uniform on $[0, 1)$.
Set $q \in [0, 1]^3$ with $q_j = x_j / \tau$ and $Q = \sum_{j=1}^3 q_j$. By \eqref{eq:x-constraints}, 
\[
    1 \le Q \le2.
\]
Since $c_j - c_{j-1} = q_j \le 1$, the events $U' \in [c_{j-1}, c_j)$ and $U'+1 \in [c_{j-1}, c_j)$ are disjoint. Hence $Z_j\in\{0,1\}$ deterministically and $\Pr(Z_j = 1) = c_j - c_{j-1} =  q_j$.
Since the intervals $[c_{j-1},c_j)$ partition $[0,Q)$, we have
\[
    \sum_{j=1}^3 Z_j
    = \mathbf{1}\{U' \in[0,Q)\} + \mathbf{1}\{U'+1\in[0,Q)\}
    \in\{1,2\},
\]
where the inclusion follows from the fact that $1 \le Q \le 2$.
Thus, $\delta  = 2Z-\mathbf{1} \in \{\pm 1\}^3$ satisfies $\sum_j \delta_j \in \{\pm 1\}$ deterministically.
In summary, we have 
\[
\Pr(\delta_j = 1 \mid E_1) = q_j = \frac{x_j}{\tau}, \quad  \Pr(\delta_j = -1 \mid E_1) = 1 - q_j = \frac{\tau-x_j}{\tau}.
\]

Now consider the branch based on the event $E_2 = \{U  \ge \tau\}$. In this case we sample $\delta\in\{\pm e_1,\pm e_2,\pm e_3\}$
according to $\Pr(\delta=e_j\mid E_2)=\frac{a_j-x_j}{1-\tau}$ and $ \Pr(\delta=-e_j\mid E_2) =\frac{b_j-\tau+x_j}{1-\tau}$.
It follows from our choice of $x$ that these probabilities are nonnegative and that their numerators sum to $\sum_{j=1}^3 (a_j-x_j+b_j-\tau+x_j) =1-\tau$.
We also have $\sum_j \delta_j \in \{\pm 1\}$ deterministically. 

Putting both cases together, we conclude that the marginals satisfy
\[
    \Pr(\delta_j=-1) = \Pr(\delta_j=-1 \mid E_1) \Pr(E_1)+ \Pr(\delta_j=-1 \mid E_2)\Pr(E_2)
    =(\tau-x_j)+(b_j-\tau+x_j)=b_j,
\]
and
\[
    \Pr(\delta_j=1) = \Pr(\delta_j=1 \mid E_1)\Pr(E_1) + \Pr(\delta_j=1 \mid E_2) \Pr(E_2)
    =x_j+(a_j-x_j)=a_j,
\]
and that the random vector $\delta \in \{-1, 0, 1\}^3$ has coordinates that sum to either $-1$ or $1$ deterministically.
\end{proof}
\begin{algorithm}[t]
\caption{\texttt{Three-Way-Coupling}}
\label{alg:three-way-coupling}
\begin{algorithmic}[1]
\Require $(a_j,b_j)_{j=1}^3$ with
$0\le a_j,b_j\le \frac13$ and $a_j+b_j\ge\frac13$
\State $\displaystyle \tau\gets(\sum_{j=1}^3(a_j+b_j)-1)/2$
\State Choose $\displaystyle x\in \mb{R}^3$ with $\max\{\tau-b_j,0\} \le x_j\le\min\{\tau,a_j\}$ and $\displaystyle \tau\le x_1 + x_2 + x_3 \le2\tau$
\State Sample $U\sim\mathrm{Uniform}[0,1)$
\If{$U<\tau$}
    \State $c_0 \gets 0$
    \For{$j=1$ to $3$}
    \State $q_j \gets x_j/\tau, \quad c_j\gets\sum_{i=1}^jq_i$
    \State $Z_j\gets\mathbf 1\!\left\{U/\tau\in[c_{j-1},c_j)\right\} + \mathbf 1\!\left\{U/\tau + 1\in[c_{j-1},c_j)\right\}$
    \EndFor
    \State $\delta \gets (2Z_1 - 1, 2Z_2 - 1, 2Z_3 - 1)$
\Else
    \State Sample $\delta\in\{\pm e_1,\pm e_2,\pm e_3\}$ with $\Pr(\delta=e_j)=\frac{a_j-x_j}{1-\tau}$ and $
    \Pr(\delta=-e_j)=\frac{b_j-\tau+x_j}{1-\tau}$
\EndIf
\State \Return $\delta$
\Comment{$\delta_1+\delta_2+\delta_3\in\{-1,1\}$}
\end{algorithmic}
\end{algorithm}

\begin{algorithm}[H]
\caption{\texttt{Gaussian-Triplet-Walk}}\label{alg:online-gaussian-coupled-coloring}
\begin{algorithmic}[1]
\Require Online vectors $v_1,\dots,v_T\in\mb{R}^d$ with $\|v_t\|_2\le 1$
\State Sample independent $G,G'\sim \mc{N}(0,I_d)$
\State Initialize
\[
    G_{0, 1} \gets G,\qquad
    G_{0, 2} \gets -\frac12G+\frac{\sqrt3}{2}G',\qquad
    G_{0, 3} \gets -\frac12G-\frac{\sqrt3}{2}G'
\]
\State $S_0 \gets 0$ \Comment{$S_0 = \sum_{j=1}^3 G_{0, j}$}
\For{$t=1$ to $T$}
    \State Receive $v_t$
    \For{$j=1$ to $3$}
        \State $p_{t,j}^+ \gets \frac13\min\{1,\exp(-\langle G_{t-1, j},v_t\rangle-\frac{\|v_t\|_2^2}{2})\}$
        \State $p_{t,j}^- \gets \frac13\min\{1,\exp(\langle G_{t-1, j},v_t\rangle-\frac{\|v_t\|_2^2}{2})\}$
    \EndFor
    \State $(\delta_{t,1}, \delta_{t,2}, \delta_{t,3})\gets
    \texttt{Three-Way-Coupling}\big((p_{t,j}^+,p_{t,j}^-)_{j=1}^3\big)$
    \Comment{\cref{alg:three-way-coupling}}
    \For{$j=1$ to $3$}
        \State $G_{t, j} \gets G_{t-1, j}+\delta_{t, j} v_t$ \Comment{$G_{t,j}$ is marginally $ \mc{N}(0, I_d)$}
    \EndFor
    \State $\varepsilon_t \gets \delta_{t,1}+\delta_{t,2}+\delta_{t,3}$
    \Comment{$\varepsilon_t\in\{-1,1\}$}
    \State $S_t \gets S_{t-1} + \eps_t v_t$ \Comment{$S_t = \sum_{j=1}^3 G_{t, j}$}
\EndFor
\State \Return $S_T$ 
\end{algorithmic}
\end{algorithm}
We now prove the main result.
\begin{proof}[Proof of \cref{thm:3gaussians}]
Notice that by construction $\sum_{j=1}^3 G_{0, j} = 0$. 
Thus, $S_0 = \sum_{j=1}^3 G_{0, j}$, so the identity $S_t=\sum_{j=1}^3G_{t,j}$ follows deterministically from the updates.
We need to prove that \(\varepsilon_t\in\{\pm1\}\) and that each
\(G_{t,j}\) remains marginally standard Gaussian at every time step $t \in [T]$.
Since $\|v_t\|_2 \le 1$, \cref{prop:nice_probs} guarantees
\begin{equation}\label{eq:final_prob_guarantee}
 p_{t, j}^+, p_{t, j}^- \le \frac{1}{3}, \quad p_{t, j}^+ +  p_{t, j}^- \ge \frac{1}{3}, 
\end{equation}
for every $t \in [T]$ and $j \in [3]$. Thus, the inputs into \texttt{Three-Way-Coupling} satisfy the requirement of~\cref{lem:three-way-coupling}, which guarantees
\[
\eps_t = \sum_{j=1}^3 \delta_{t, j} \in \{\pm 1\},
\]
deterministically.
Next, we prove by induction over $t \in [T]$ that $G_{t, j} \sim \mc{N}(0, I_d)$ for every $j \in [3]$.
The base case $t=0$ holds by construction, since $G_{0,j}\sim\mc{N}(0,I_d)$.
Assume that the inductive hypothesis holds at time step $t-1$. 
We will show that it holds at time step $t$.
Fix $j \in [3]$. We have 
\[
G_{t, j} = G_{t-1, j} + \delta_{t, j} v_t,
\]
where $G_{t-1, j} \sim \mc{N}(0, I_d)$. Since \cref{eq:final_prob_guarantee} holds, \cref{lem:three-way-coupling} yields
\begin{equation}\label{eq:delta_distribution}
\begin{split}
&\Pr(\delta_{t, j} = -1 \mid G_{t-1, j}) = p^-_{v_t}(G_{t-1, j}), \quad
\Pr(\delta_{t, j} = 1 \mid G_{t-1, j}) = p^+_{v_t}(G_{t-1, j}),
\end{split}
\end{equation}
so \cref{prop:gaussian_markov_chain} guarantees $G_{t, j} \sim \mc{N}(0, I_d)$. Since $j \in [3]$ was arbitrary, this completes the induction.

We now address the running time. 
Sampling $G,G' \sim \mathcal N(0,I_d)$ and initializing $(G_{0,j})_{j=1}^3$ takes $O(d)$ time.
At every time step it takes $O(d)$ time to compute $\|v_t\|_2^2$ and the three inner products.
It takes $O(1)$ time to sample $(\delta_{t,1}, \delta_{t,2}, \delta_{t,3})$ using the construction in the proof of \cref{lem:three-way-coupling}.
Finally, updating the auxiliary walks $(G_{t, j})_{j=1}^3$ and the main walk $S_t$ can be done in $O(d)$ time.
Thus, the total running time is $O(dT)$.
\end{proof}
\section{Discrepancy minimization}\label{sec:discrepancy}
In this section we prove \cref{thm:prefix_discrepancy,cor:offline_komlos}.
We begin with a simple proof that the Gaussian triplet walk has small prefix discrepancy.
\begin{proof}[Proof of \cref{thm:prefix_discrepancy}]
Fix $t \in [T]$.
Let $E_t \subseteq \mb{R}^d$ be the subspace spanned by $\{v_1, \dots, v_t\}$ and let $P_t$ be the orthogonal projection matrix onto $E_t$. 
Let $d_t \le \min\{d, T\}$ be the dimension of $E_t$. 
Since $S_t \in E_t$, for any $\theta \in \mb{R}^d$ we have $\la S_t , \theta \ra = \la S_t , P_t \theta \ra$.
Thus, \cref{thm:3gaussians} and H\"older's inequality give, for any $\theta \in \mb{R}^d$,
\begin{equation}\label{eq:holder}
\E\left[\exp\left(\la S_t , \theta \ra \right) \right] = \E\left[\exp\left(\la S_t , P_t\theta \ra \right) \right] \le \prod_{j=1}^3\E\left[\exp\left(\la G_{t, j}, 3P_t\theta \ra \right) \right]^{1/3} = \exp\left(\frac{9\|P_t\theta\|_2^2}{2}\right).
\end{equation}
Set $V = \{ k : P_t e_k \not= 0\}$ and fix $u > \sqrt{18}$. 
The union bound, Markov's inequality, and \cref{eq:holder} yield
\begin{align*}
\Pr\left(\|S_t\|_\infty > u \right) &\le  \sum_{k =1}^d \inf_{\lambda > 0}\Pr\left(\exp\left(\lambda \la S_t, P_te_k \ra \right) > \exp(\lambda u)\right) + \sum_{k =1}^d \inf_{\lambda > 0}\Pr\left(\exp\left(\lambda \la S_t, -P_te_k \ra \right) > \exp(\lambda u)\right) \\
&\le 2\sum_{k \in V} \inf_{\lambda > 0} \exp\left(-\lambda u + \frac{9\lambda^2\|P_te_k\|_2^2}{2}\right)\\
&=2\sum_{k \in V} \exp\left( -\frac{u^2}{18\|P_te_k\|_2^2}\right).
\end{align*}
Furthermore,
\[
\Pr\left(\|S_t\|_\infty > u \right) \le 2 \sum_{k \in V} \|P_t e_k\|_2^2 \exp\left( - \frac{u^2}{18}\right) = 2 d_t \exp\left( - \frac{u^2}{18}\right) ,
\]
where we used the fact that $ \sum_{k = 1}^d \|P_t e_k\|_2^2 = \mathrm{Tr}(P_t) = d_t$ and the inequality $\exp(-x/a) \le a \exp(-x)$, valid for $ 0 < a \le 1$ and $ x \ge 1$.
Applying the union bound over all $t\in [T]$ gives us
\[
\Pr\left( \max_{t \in [T] } \|S_t\|_\infty > u \right) \le 2\sum_{t =1}^T d_t \exp\left( - \frac{u^2}{18}\right) \le 2T^2 \exp\left(-\frac{u^2}{18}\right). 
\]
Setting $u = \sqrt{18 \log \left( 2T^2/\delta \right) }$ yields
\[
\Pr\left( \max_{t \in [T] } \|S_t\|_\infty > \sqrt{18 \log \left( \frac{2T^2}{\delta} \right) } \right) \le \delta,
\]
as claimed.
\end{proof}
For completeness, we also include a simple proof that the final step of the Gaussian triplet walk incurs discrepancy matching Banaszczyk's bound for the Koml\'os problem.
\begin{proof}[Proof of \cref{cor:offline_komlos}]
The algorithm feeds the columns of $A = (v_1, \dots, v_T) \in \mb{R}^{d \times T}$ into \cref{alg:online-gaussian-coupled-coloring} in an online fashion.
We modify the algorithm to return $\eps = (\eps_1, \dots, \eps_T)$ instead of $S_T$.
Using the same argument as in the proof of \cref{thm:prefix_discrepancy}, we get that $A\eps = S_T$ satisfies
\[
\Pr\left( \|A\eps\|_\infty > u  \right)  \le 2 \min\{d, T\} \exp \left( -\frac{u^2}{18}\right),
\]
which holds for any $u \ge \sqrt{18}$.
By setting $u = \sqrt{18\ln(2\min\{d, T\}/\delta)}$ we get
\[
\Pr\left( \|A\eps\|_\infty > \sqrt{18 \log \left( \frac{2\min\{d, T\}}{\delta} \right) } \right) \le \delta,
\]
as claimed.
\end{proof}

\paragraph{Acknowledgments.} I thank Sidhanth Mohanty for helpful comments that improved the presentation.

\footnotesize
\bibliographystyle{alpha}
\bibliography{refs.bib}

\end{document}

%% file: macros.tex
\usepackage{amsmath,amsthm,amssymb}
\usepackage{textcomp}
\usepackage{mathtools,bbm,xspace}
\usepackage{mathrsfs}
\usepackage{nicefrac}
\usepackage[numbers]{natbib}
\usepackage[dvipsnames]{xcolor}
\usepackage{pgfplots}
\pgfplotsset{compat=1.18}
\usepackage{ifxetex}
\usepackage{dirtytalk}
\usepackage{soul}

\usepackage[pagebackref]{hyperref}
\hypersetup{
colorlinks=true,
urlcolor=blue,
linkcolor=blue,
citecolor=OliveGreen,
}

\usepackage[capitalise,nameinlink]{cleveref}
\crefname{proposition}{Proposition}{Propositions}
\crefname{lemma}{Lemma}{Lemmas}
\crefname{fact}{Fact}{Facts}
\crefname{theorem}{Theorem}{Theorems}
\crefname{corollary}{Corollary}{Corollaries}
\crefname{conjecture}{Conjecture}{Conjectures}
\crefname{claim}{Claim}{Claims}
\crefname{example}{Example}{Examples}
\crefname{problem}{Problem}{Problems}
\crefname{setting}{Setting}{Settings}
\crefname{definition}{Definition}{Definitions}
\crefname{assumption}{Assumption}{Assumptions}
\crefname{subsection}{Subsection}{Subsections}
\crefname{section}{Section}{Sections}

\usepackage{tcolorbox}
\usepackage{enumerate}
\usepackage{csquotes}

\usepackage{algorithm}
\usepackage{algpseudocode}
\usepackage{tabularx}

\usepackage{todonotes}

\usepackage{bm}
\usepackage{aliascnt}
\allowdisplaybreaks

\newtheorem{theorem}{Theorem}
\newtheorem*{theorem*}{Theorem}

\newaliascnt{proposition}{theorem}
\newtheorem{proposition}[proposition]{Proposition}
\aliascntresetthe{proposition}
\newtheorem*{proposition*}{Proposition}
\newaliascnt{lemma}{theorem}
\newtheorem{lemma}[lemma]{Lemma}
\aliascntresetthe{lemma}
\newtheorem*{lemma*}{Lemma}
\newaliascnt{corollary}{theorem}
\newtheorem{corollary}[corollary]{Corollary}
\aliascntresetthe{corollary}
\newtheorem*{corollary*}{Corollary}
\newtheorem*{conjecture*}{Conjecture}

\newtheorem*{fact*}{Fact}

\newtheorem*{exercise*}{Exercise}

\newtheorem*{hypothesis*}{Hypothesis}

\usepackage{thm-restate}
\usepackage{thm-restate}

\theoremstyle{definition}

\newtheorem{exercise-easy}[theorem]{Exercise}
\newtheorem{exercise-med}[theorem]{Exercise}
\newtheorem{exercise-hard}[theorem]{Exercise$^\star$}

\newtheorem*{claim*}{Claim}

\newtheorem*{remark*}{Remark}

\newtheorem*{observation*}{Observation}

\newcommand{\Sref}[1]{\hyperref[#1]{\S\ref*{#1}}}

\usepackage[
letterpaper,
top=1in,
bottom=1in,
left=1in,
right=1in]{geometry}

\DeclarePairedDelimiterX{\norm}[1]{\lVert}{\rVert}{#1}

\DeclarePairedDelimiterX{\abs}[1]{\lvert}{\rvert}{#1}
\DeclarePairedDelimiterX{\inp}[2]{\langle}{\rangle}{#1, #2}

\DeclarePairedDelimiterX{\infdivx}[2]{(}{)}{%
  #1\;\delimsize\|\;#2%
}

\DeclareMathOperator*{\E}{\mathbb{E}}

\let\Pr\relax
\DeclareMathOperator*{\Pr}{\mathbf{Pr}}

\newcommand{\mc}[1]{\mathcal{#1}}
\newcommand{\mb}[1]{\mathbb{#1}}

\newcommand{\eps}{\varepsilon}

\newcommand{\la}{\langle}
\newcommand{\ra}{\rangle}

\def\com{1}
\newcommand{\ishaq}[1]{
    \if\com1
        \todo[inline,color=blue!30]{\small\textbf{Ishaq:} #1}
    \else
    \fi
}

\crefformat{equation}{(#2#1#3)}